\documentclass[11pt]{article}
\usepackage{fullpage}

\usepackage{amssymb}
\usepackage{graphicx}
\usepackage{url}
\usepackage{nicefrac}
\usepackage{enumitem}
\usepackage{authblk}
\usepackage{tikz} 
\usepackage{subcaption}
\usepackage{xspace}
\usepackage[compact]{titlesec}
\usetikzlibrary{arrows,decorations.markings}
\usepackage[normalem]{ulem}
\usepackage[utf8]{inputenc}
\usepackage{amsmath}
\usepackage{amsthm}
\usepackage{algorithm}
\usepackage[noend]{algpseudocode}
\usepackage{hyperref}
\usepackage{fancyvrb,xcolor}
\usepackage{tikz}
\usetikzlibrary{shapes,shadows,arrows}

{\bfseries}{\itshape}

 \newcommand{\ts}{\mathsf{timestamp}}
 \newcommand{\pr}{\mathsf{price}}
 \newcommand{\Q}{\mathsf{Qty}}
 \newcommand{\q}{\mathsf{qty}}
 \newcommand{\id}{\mathsf{id}}
 \newcommand{\Bids}{\mathsf{Bids}}
 \newcommand{\Asks}{\mathsf{Asks}}
 \newcommand{\idb}{\mathsf{id_{bid}}}
 \newcommand{\ida}{\mathsf{id_{ask}}}
 \newcommand{\idsb}{\mathsf{ids_{bid}}}
 \newcommand{\idsa}{\mathsf{ids_{ask}}}
 \newcommand{\ids}{\mathsf{ids}}
 \newcommand{\vol}{\mathsf{Vol}}
 
 \newcommand{\cmd}{\mathsf{command}}
 \newcommand{\ins}{\mathsf{instruction}}
 \newcommand{\buy}{\mathsf{Buy}}
 \newcommand{\sell}{\mathsf{Sell}}
 \newcommand{\del}{\mathsf{Del}}
 \newcommand{\absorb}{\mathsf{Absorb}}
 \newcommand{\iter}{\mathsf{Iterated}}
 \newcommand{\ip}{\mathsf{Iterated\text{-}P}}
\newcommand{\matchask}{\mathsf{Match\_Ask}}
\newcommand{\matchbid}{\mathsf{Match\_Bid}}
\newcommand{\deleteorder}{\mathsf{Del\_Order}}

\newtheorem{theorem}{Theorem}
\newtheorem{proposition}{Proposition}
\newtheorem{lemma}{Lemma}
\usepackage[T1]{fontenc}
\usepackage{calc}
\usepackage{color}
\usepackage{lipsum}

\makeatother
\usepackage{color}

\definecolor{blueblack}{rgb}{0,0,.7}

\newcounter{sideremark}

\definecolor{Darkblue}{rgb}{0,0,0.4}
\definecolor{Brown}{cmyk}{0,0.61,1.,0.60}
\definecolor{Purple}{cmyk}{0.45,0.86,0,0}
\definecolor{brickred}{rgb}{0.8, 0.25, 0.33}

\title{The Design and Regulation of Exchanges: A~Formal Approach} 
\author[1]{Mohit Garg}
\author[2]{Suneel Sarswat}
\affil[1]{University of Bremen, University of Hamburg, Germany, garg@uni-bremen.de}
\affil[2]{Tata Institute of Fundamental Research, India, suneel.sarswat@gmail.com}

\begin{document}

\maketitle

\begin{abstract}
We use formal methods to specify, design, and monitor continuous double auctions, which are widely used to match buyers and sellers at  exchanges of foreign currencies, stocks, and commodities. We identify three natural properties of such auctions and formally prove that these properties completely determine the input-output relationship. We then formally verify that a natural algorithm satisfies these properties. All definitions, theorems, and proofs are formalized in an interactive theorem prover. We extract a verified program of our algorithm to build an automated checker that is guaranteed to detect errors in the trade logs of exchanges if they generate transactions that violate any of the natural properties.
\end{abstract}

\maketitle

\section{Introduction}
Continuous double auctions are widely used to match buyers and sellers at market institutions such as foreign exchange markets, cryptocurrency exchanges, stock exchanges, and commodities exchanges. 
Increasingly, the exchanges deploy automated computer systems for conducting trades, which often receive a huge number of trade requests, especially in presence of high-frequency algorithmic traders. A minor bug in the computer system of a major exchange can have a catastrophic effect on the overall economy. To ensure that the exchanges work in a fair and orderly manner, they are subject to various regulatory guidelines. There have been a number of instances where the exchanges have been found violating regulatory guidelines \cite{nyse1,nyse2,nse,ubs}. For example, it was observed that NYSE Arca failed to execute certain trades in violation of the stated rules \cite{nyse1}, which led to the U.S. Securities and Exchange Commission imposing a heavy penalty on the New York Stock Exchange. 

The above example is an instance of a program not meeting its specification, where the exchange's order matching algorithm is the program and the exchange rules and regulatory guidelines form the broad specifications for the program.
The reasons for such violations are generally twofold. Firstly, the exchange rules and the regulatory guidelines are not presented in a formal language; thus, they tend to be rather vague and ambiguous. Even if they are stated unambiguously, it is not clear whether they are consistent; they lack a formal proof of consistency. Secondly, the fidelity of the program implementing the specifications is traditionally based on repeatedly testing the software on large data sets and removing all bugs that get detected. Although this approach irons out many bugs, it is not foolproof; in particular, many bugs that surface only in rare scenarios may easily remain hidden in the program.
A related concern that is often missed is that a faulty program might continue to make mistakes without anyone noticing or even realizing that they are mistakes. This is especially relevant in the context of exchanges, where individual traders have a limited view of the system and may not realize that the exchange is making mistakes, and the regulators might overlook them, probably because they are not the kinds of mistakes that they look for or are brought to their notice.

In this work, we introduce a new framework for exchange design and regulation that addresses the above concerns. We elucidate our approach for an exchange that implements the widely used continuous double auction mechanism to match buy and sell requests. To this end, we develop the necessary theoretical results for continuous double auctions which then allows us to apply the formal methods technology, which has classically been applied to safety-critical systems like nuclear power plants~\cite{nuclearformal}, railway signaling systems, flight control software, and in hardware design (see \cite{formalpractice} for a survey).

More specifically, we consider a general model for continuous double auctions and identify three natural properties, namely \texttt{price-time priority}, \texttt{positive bid-ask spread}, and \texttt{conservation}, that are necessary for any online algorithm implementing continuous double auction to possess. We then show that any algorithm satisfying these three properties for each input must have a `unique' output at every time step, implying that these three properties are sufficient for specifying continuous double auctions, and there is no need to impose any further requirements on an algorithm implementing such an auction. We then develop an algorithm and formally prove that it satisfies these three properties. This establishes the consistency of the three properties. 

All our proofs are machine-checked; all the notions, definitions, theorems, and proofs are formalized in the popular Coq proof assistant~\cite{coq}. Finally, using Coq's program extraction feature, we obtain an OCaml program of our verified algorithm for implementing continuous double auctions. Such a verified program is guaranteed to be bug-free and can be directly used at an exchange.  Additionally, enabled by the uniqueness theorem we prove, a verified program can be used to check for errors automatically in an existing exchange program by comparing the outputs of the verified program against that of the exchange program. Such a verified checker goes over the trade logs of the exchange and is guaranteed to detect a violation if the exchange output violates any of the three properties, and this alleviates the final concern we mentioned above. This can be a crucial tool for regulators. As an application, using our verified program we check for errors in real data collected from an exchange.

Thus, our approach can be summed up as follows: The specification for an exchange should consist of a set of provably consistent properties, which are preferably few and simple. Furthermore, these properties should be rich enough so that one can build an automated checker that can always detect a violation of the specification by going over the logs of the exchange if one exists. 

We begin by briefly describing the trading process using continuous double auctions. 
\subsection{Overview of continuous double auctions}
A double auction is a process to match multiple buy and sell orders for a given product at a marketplace. Potential buyers and sellers of a product place their orders at the market institution for trade. 
Buy and sell orders are referred to as bids and asks, respectively. The market institution on receiving orders is supposed to produce transactions.  Each transaction is between a bid and an ask and consists of a transaction price and a transaction quantity.

More specifically, each order consists of an order id, timestamp, limit price, and maximum quantity. 
Order ids are used to distinguish orders (and hence are unique for each order). Timestamps represent the arrival times of the orders. For a fixed product, the timestamps of the orders are distinct.\footnote{In real systems, even if multiple orders arrive at the same time, they are entered into the system one by one. How this is exactly achieved is beyond the scope of this work.}
The limit price of an order $\omega$ is the maximum (if $\omega$ is a bid) or minimum (if $\omega$ is an ask) possible transaction price for each transaction involving $\omega$. The maximum quantity $q$ of an order $\omega$ is the number of units of the product offered for trade, i.e., the sum of the transaction quantities, where $\omega$ participates in, is at most $q$.

 Double auctions are used to generate transactions systematically.
There are two common types of double auctions that are used at various exchanges: call auctions~\cite{WWW98,NiuP13,zhao2010maximal} and continuous double auctions. We discuss call auctions in Section~\ref{related}.
 In a system implementing continuous double auctions, buyers and sellers may place their orders at any point in time. On receiving an order $\omega$, the system instantly tries to match $\omega$ with fully or partially unmatched orders that arrived earlier (`resident orders') and output transactions, before it proceeds to process any other incoming orders. If there are multiple orders with which $\omega$ can be matched, then the system prioritizes those matchable orders based on \texttt{price-time priority} (first the orders are ranked by the competitiveness of their price;  orders with the same price are then ranked by their arrival time, i.e., as per their timestamps). If the total transaction quantities of the transactions generated by $\omega$ exhaust the maximum quantity of $\omega$, then it leaves the system completely. Else $\omega$ with its remaining unmatched quantity is retained by the system for potential future matches; we refer to such unmatched orders as resident orders. A resident order leaves the system when its quantity gets fully exhausted.
 
The general model we work with has three primitive instruction types: buy, sell, and delete.\footnote{In Section~\ref{sec:application} we describe how certain other instruction types can be converted to these primitives.} That is, apart from buy and sell instructions, it is possible to delete a resident order by giving a `delete id' instruction to the system. %
The system maintains two main logbooks: an order book and a trade book. The order book is a list of all buy, sell, and delete instructions that the system received, ordered by their timestamps. The trade book consists of the list of all transactions generated by the system in the order in which they were generated. 
 
 We now describe three natural properties of a system implementing 
 continuous double auctions that are relevant to our formalization.

 \begin{itemize}
  \item{\texttt{Positive bid-ask spread.}} This property represents the `inertness' of resident orders; at any point in time, the resident orders are not matchable to each other, i.e., the limit price of each resident bid is strictly less than the limit price of each resident ask. This follows immediately from the fact that an order becomes resident only when it cannot be matched any further with the existing resident orders (or when there are no matchable resident orders at all).
 
 \item{\texttt{Price-time priority.}} When a new order arrives, the system tries to match it with the resident orders and generate transactions. If a less-competitive resident order participates in one of these transactions, then all current resident orders which are more competitive must get fully traded in these transactions, where competitiveness is decided based on 
 \texttt{price-time priority}. This is a direct consequence of the matching process described above.

 \item{\texttt{Conservation.}} This property expresses the `honesty' of the system. Roughly speaking, it states that the system should not lose, create, or modify orders arbitrarily. For example, when a new incoming order gets partially traded with existing resident orders, it becomes resident with its maximum quantity precisely being the difference between its original maximum quantity and the total quantity of the transactions that are generated. Furthermore, the timestamp, id, and limit price of the resident order must remain unchanged. 
  \end{itemize}

One might feel that many properties would be necessary for specifying `conservation' formally. On the contrary, later we will see that the above three properties can be stated rather succinctly once we set up our definitions appropriately. Also, we will show that any algorithm that implements these three properties, on any input (order book), will produce a `unique' output at every point in time, implying that these properties completely characterize continuous double auctions.
 One can learn more about double auctions from~\cite{harris2003trading,ptpandprorata,xetra,D93}.
 
 \subsection{Results}
 After setting up the definitions appropriately, we formally represent the three natural properties of continuous double auctions that we discussed above, namely \texttt{price-time priority}, \texttt{positive bid-ask spread}, and \texttt{conservation}. We are then able to prove the following results.

 \begin{itemize}
     \item{Maximum matching.} We first show in Lemma~\ref{lem:max} that any algorithm that satisfies \texttt{positive bid-ask spread} and \texttt{conservation}, for each incoming order will generate a set of transactions (`matching') whose total quantity will be maximum, i.e., it will be at least the total quantity of any other feasible set of transactions.
     
     \item{Local Uniqueness.} In Theorem~\ref{thm:processUniqueness} we show that any two processes that satisfy the aforementioned three properties, for any set of resident orders and an incoming instruction that satisfy certain technical conditions, will essentially generate the same matching, and the resident orders that remain at the end of processing this instruction will also be the same. To prove this theorem we make use of Lemma~\ref{lem:max}.
     
     \item{Global Uniqueness.} We then show in Theorem~\ref{thm:globalUniquness} that any two processes that satisfy the three properties, given any order book 
     as input, at any point in time will essentially generate the same matching, and the resident orders that remain will also be the same. We prove this theorem by lifting Theorem~\ref{thm:processUniqueness} via an induction argument.
     
     \item{Correctness of $\mathsf{Process\_instruction}$}. We then design an algorithm for continuous double auctions which we refer to as $\mathsf{Process\_instruction}$ and then in Theorem~\ref{thm:Process_instructionSatisfiesProperties} show that it satisfies the three properties.

 \end{itemize}

 To use the above results in a practical setting,
 we extract an OCaml program of our verified algorithm $\mathsf{Process\_instruction}$, using Coq's code extraction feature, and then use it to check for errors  in real data from an exchange by running it on order books and then comparing the outputs with the corresponding trade books. Our global uniqueness theorem implies that if the exchange program has the three natural properties, then the output of our program should `match' with the output of the exchange. If they do not match, then at least one of the three properties must be violated by the exchange program.
 
     Our Coq formalization uses about 6000 lines of new code, which includes about 250 lemmas and 80 definitions and functions, and is available in the accompanying supplementary materials~\cite{suppl}.
 
 \subsection{Related work} 
 \label{related} 
 Mavroudis and Melton~\cite{libra} study fairness, transparency, and manipulation in exchange systems. They argue that idealized market models do not capture infrastructural inefficiencies and their simple mechanisms are not robust enough to withstand sophisticated technical manipulation attacks; as a result, the theoretical fairness guarantees provided by such exchanges are not retained in reality. To address this problem, they propose the LIBRA system that achieves a relaxed notion of fairness by making clever use of randomness in routing orders to the matching engine. Our work focuses on a complementary and an arguably more basic aspect of exchanges; once the system assigns unique timestamps, even if they are `inaccurate' or `unfair' due to infrastructural inaccuracies or other reasons, with respect to those timestamps the matching engine should respect the desired properties; our work focuses on specifying and checking the correctness of the matching engine.

 There have been some important works on formalizing financial systems, double auctions, and theorems in economics. We briefly mention a few of them. In an influential work \cite{PI17}, Passmore and Ignatovich highlight the need for formal verification of financial algorithms and suggest various open problems in the field. In response to these challenges, they designed Imandra~\cite{imandra}, a specialized formal verification system and programming language designed to reason about properties of algorithms that may be proved, refuted, or described.
 
 Cervesato, Khan, Reis, and \v{Z}uni\'c~\cite{clf} present a declarative and modular specification for an automated trading system for continuous double auctions in a concurrent linear framework (CLF) and implemented it in a CLF type checker that also supports executing of CLF specifications. Among other things, they were able to establish that their system is never in a locked or crossed state, which is equivalent to the \texttt{positive bid-ask spread} property mentioned above.
 
 Wurman, Walsh, and Wellman \cite{WWW98}  deal with the theory and implementation of flexible double auctions for electronic exchanges. In particular, the authors analyze the incentive compatibility of such auctions.  Kaliszyk and Parsert~\cite{KP18} introduce a formal micro-economic framework and formally prove the first welfare theorem leading to a more refined understanding of the theorem than was available in the economics literature. 

{\bf Comparison with call auctions.} Our work is closely related to the work of Raja, Singh, and Sarswat~\cite{RSS21}, which formalizes call auctions. As mentioned earlier, call auctions form a class of double auctions, which are `non-continuous' or `one-shot' in nature. In call auctions, orders are collected from buyers and sellers for a fixed duration of time, at the end of which the orders are matched simultaneously to generate transactions. 
In \cite{RSS21} call auctions are formalized and certain uniqueness theorems are obtained. Unlike call auctions, prior to our work, for continuous double auctions the specifications were not concisely stated in the literature, hindering formal development as also noted in \cite{RSS21}.

It is interesting to compare call auctions with one time instant of continuous double auctions when a new order arrives in presence of resident orders. 
In call auctions two different objectives are usually considered: 1. Maximum matching: produce transactions that maximize the trade volume (the sum of transaction quantities); 2. Maximum uniform-price matching: produce transactions that maximize the trade volume subject to the constraint that all the transaction prices are the same. It turns out, there are instances where meeting these objectives leads to different matchings and trade volumes. In either case, after the transactions are produced, the orders or parts of the order that remain unmatched are not-matchable, i.e., they have \texttt{positive bid-ask spread}. 
 Compare this with one time instant of continuous double auction. As we show in this work, \texttt{positive bid-ask spread} already makes the process unique and can be seen as a special case of a call auction where the output matching meets both the objectives (1 and 2) simultaneously.\footnote{For this, we need a slightly relaxed definition of uniformity: we say a matching has a uniform price if there exists a common price that can be set for all transactions in the matching such that no limit price of the participating orders is breached.} In the context of continuous double auctions, this leads to an arguably much simpler requirement of \texttt{positive bid-ask spread} as opposed to the requirements of maximum matching or maximum uniform-price matching. 
 Additionally, we derive a stronger uniqueness theorem for continuous double auctions: for each bid-ask pair, the total transaction quantity between them is unique; in contrast, for call auctions, for each order, the total transaction quantity involving that order is unique.
 Where continuous double auctions completely differ from call auctions is in their online setting, which adds a new layer of complexity to the formalization.
 
 \subsection{Organization of the rest of the paper} For understanding the rest of the paper, we do not assume that the reader has any prior knowledge or expertise in formalization or using interactive theorem provers like Coq. For an interested reader, we have provided in~\cite{suppl} details  that show how the notions presented here are represented in our Coq formalization, and one can convince oneself of the correctness of the results by simply compiling the accompanying Coq formalization~\cite{suppl} without having to read the proofs presented here.
 Additionally, a demonstration is included in the supplementary materials; one needs an OCaml compiler to be able to run the demonstration.
 
 In Section~\ref{sec:preliminaries} we present the various notions and definitions that are needed for our work.
Section~\ref{sec:properties} has the three natural properties stated formally. In Sections~\ref{sec:maximum}-\ref{sec:verified}, we prove the above mentioned results: Maximum matching, Local Uniqueness, Global Uniqueness, and Correctness of $\mathsf{Process\_instruction}$ (partly moved to Appendix~\ref{sec:algoCorrectness}. 
In Section~\ref{sec:application} we describe how to build automated checkers. 
We also include a demonstration to show how our checker works on real data. Finally, the last section concludes the paper and describes future directions.

\section{Preliminaries}\label{sec:preliminaries}
We begin by introducing the various definitions that are needed for establishing our results. For ease of readability, we extensively use sets\footnote{Sets introduced in this work are all finite.} in our presentation. In the Coq formalization, however, we use lists instead. The choice of lists in the formalization allows us to use existing libraries that were developed in \cite{SS20} and \cite{RSS21} for the purposes of modeling auctions, and, more importantly, it helps us in optimizing our algorithm leading to a reasonably fast OCaml program.\footnote{Using other data structures might yield further improvements.} In~\cite{suppl}, we include detailed definitions and the main results alongside their formal versions, showing a faithful correspondence between the results and the formalization.

\subsection{Orders}

To avoid proving common properties of bids and asks twice, we model both as orders. 
An order is a $4$-tuple $(id$, $timestamp$, $quantity$, $price)$. For an order $\omega$, its components are $\id(\omega)$, $\ts(\omega)$, $\q(\omega)$, and $\pr(\omega)$, each of which is a natural number, and $\q(\omega)>0$.

For a set of orders $\Omega$, $\ids(\Omega)$ represents the set of ids of the orders in $\Omega$. For a set of orders $\Omega$ with distinct ids and an order $\omega\in\Omega$ such that $\id(\omega)=id$, with slight abuse of notation, we define $\ts(\Omega,id)=\ts(\omega)$, $\q(\Omega,id)=\q(\omega)$, and $\pr(\Omega, id)=\pr(\omega)$.

We will often have a universe from which the bids and asks arise, which we call an order-domain. $(B,A)$ is an {\bf order-domain}, if $B$ and $A$ are sets of orders. Here, $B$ represents a set of bids and $A$ represents a set of asks. An order-domain  where each id is distinct and each timestamp is distinct is called {\bf admissible}.

We now define the terms `tradable' and `matchable'. Given two orders $b$ (bid) and $a$ (ask), we say $b$ and $a$ are {\bf tradable} if $\pr(b)\geq \pr(a)$.
An order-domain is {\bf matchable} if it contains a bid and an ask that are tradable.

Next, we define competitiveness.
A bid $b_1$ is more {\bf competitive} compared to another bid $b_2$, denoted by $b_1 \succ b_2$, if
$\pr(b_1)$ $>$ $\pr(b_2)$ OR
     $(\pr(b_1)$ $=$ $\pr(b_2)$ AND $\ts(b_1)$ $<$ $\ts(b_2))$. 
  Similarly, an ask $a_1$ is considered more {\bf competitive} compared to another ask $a_2$, denoted by $a_1\succ a_2$, if $\pr(a_1)$ $<$ $\pr(a_2)$ OR
     $(\pr(a_1)$ $=$ $\pr(a_2)$ AND $\ts(a_1)$ $<$ $\ts(a_2))$.

We treat a set of orders also as a multiset where we suppress the quantity field of each order and set its multiplicity equal to its quantity. This view will help us succinctly state the \texttt{conservation} property and ease the formalization substantially. For a set of orders $S_1$ and $S_2$, $S_1 \setminus S_2$, represents the usual set difference, whereas $S_1 - S_2$ represents the usual multiset difference between the two sets.

\subsection{Transactions and matchings}

For our purposes, keeping the price and timestamp in a transaction is redundant, as they can be derived from the participating bid and ask.
Based on the application, a transaction between tradable orders $b$ and $a$ can be assigned an appropriate transaction price in the interval $[\pr(a),\pr(b)]$. For example, for the sake of concreteness, one may assume that such a transaction has price $\pr(a)$ and timestamp $\max\{\ts(b),\ts(a)\}$.

A transaction is a $3$-tuple $(id_b,id_a,quantity)$ of natural numbers where $id_b$ and $id_a$ represent the ids of the participating bid and ask, respectively, and $quantity >0$. For a transaction $t$ its components are represented by $\idb(t)$, $\ida(t)$, and $\q(t)$.

Let $T$ be a set of transactions. We define $\idsb(T)$ and $\idsa(T)$ to be the sets of ids of the participating bids and asks, respectively. Next, we define $\Q_{\text{bid}}(T, id_b)$ to be sum of the transaction quantities of all transactions in $T$ whose bid id is $id_b$, and $\Q_{\text{ask}}(T,id_a)$ to be the sum of the transaction quantities of all transactions in $T$ whose ask id is $id_a$. For ease of readability, we often just use $\Q$ instead of $\Q_{\text{ask}}$ and $\Q_{\text{bid}}$. We define $\vol(T)$ to be the sum of the transaction quantities of all transactions in $T$.

We say a transaction $t$ is  {\bf over} the order-domain $(B,A)$ iff $\idb(t)$ $=$ $\id(b)$ for some $b\in B$ and $\ida(t)=\id(a)$ for some $a\in A$. We say that a transaction $t$ is {\bf valid} w.r.t order-domain $(B,A)$ iff there exists $b\in B$ and $a\in A$ such that 
            (i) $\idb(t)=\id(b)$ and  $\ida(t)=\id(a)$,
            (ii) $b$ and $a$ are tradable, and
           (iii) $\q(t) \leq \min(\q(b),\q(a))$.
We say that a set of transactions $T$ is {\bf valid} over an order-domain $(B,A)$ if each transaction  in $T$ is valid over $(B,A)$.

Given a set of transactions, we would like to extract out the set of `traded' bids/asks. To this end, we define  the functions $\Bids$ and $\Asks$  as follows. Let $T$ be a set of transactions  over an admissible order-domain $(B,A)$. We define $\Bids(T, B)$ to be the set of all the bids in $B$ that participate in $T$ where the quantity of each bid $b$ is set to the sum of the transaction quantities of the transactions in $T$ involving $b$. Similarly, $\Asks(T,A)$ is the set of asks in $A$ that participate in $T$ where the quantity of each ask is set to its total traded quantity in $T$. We often simply write $\Bids(T)$ and $\Asks(T)$  instead of  $\Bids(T,B)$ and $\Asks(T,A)$ whenever $B$ and $A$ are clear from the context.

Finally, we define a matching, which is a set of transactions that can simultaneously arise from an order-domain. We also define the canonical form of a set of transactions, which will be often applied to matchings.
We say a set of valid transactions $M$ over an admissible order-domain $(B,A)$ is a {\bf matching} over $(B,A)$ if for each order $\omega\in B\cup A$, 
$\Q(M,\id(\omega))\leq \q(\omega)$.
We define the {\bf canonical form} of a set of transactions $M$, denoted by $\mathcal C(M)$ to be a set of transactions satisfying the following two properties. For each bid-ask pair $(b,a)$, $\mathcal C(M)$ consists of at most one transaction involving $b$ and $a$. Furthermore, for each bid-ask pair $(b,a)$, the total transaction quantity between them in $M$ is equal to the total transaction quantity between them in $\mathcal C(M)$. It is easy to see that $\mathcal C(M)$ always exists.

\subsection{Order book and process}

We now formally define instructions and order books. $\buy$, $\sell$, and $\del$ are called $\cmd$s.
An $\ins$ is a pair $(\Delta,\omega)$ where $\Delta$ is a $\cmd$ and $\omega$ is an order.
For convenience we represent an $\ins$  $(\Delta,\omega)$ by $\Delta \ \omega$. Also sometimes we represent $(\del, \omega)$ $\ins$ simply by $\del \ \id(\omega)$; this is done because for a $\del$ $\cmd$, only the $\id$ of the order matters. An {\bf order book} is a list where each entry is an $\ins$.

We now define a `structured' order book, which satisfies two conditions. Firstly, the timestamps of the orders must be increasing. Secondly, the ids of the orders in the non-$\del$ $\ins$s in the order book must be all distinct. We will relax the second condition slightly which will help us in certain applications. We allow an order to have an id identical only to the id appearing in an immediately preceding $\del$ order instruction. This will later help us in our application to implement an `update' instruction, by replacing it with a delete instruction followed by a $\buy$ or $\sell$ instruction carrying the same id.
Formally, an order book $\mathcal{I}=[(\Delta_0,\omega_0), \cdots,(\Delta_n,\omega_n)]$ is called {\bf structured} if the following conditions hold.
\begin{itemize}
    \item For all $i\in \{0,\cdots,n-1\},$ $\ts(\omega_i) < \ts(\omega_{i+1})$.
    \item For all $i\in \{0,1,\cdots ,n\}$, at least one of the following three conditions hold.
   (i) $\Delta_i = \del$.
    (ii) $\id(\omega_i) \notin \ids\{\omega_0,\cdots, \omega_{i-1}\}$.
    (iii) $\id(\omega_i) = \id(\omega_{i-1})$ and $\Delta_{i-1}=\del$.
\end{itemize}

We now define a process, which represents an abstract online algorithm that will be fed resident bids and asks and an instruction and it will output a set of transactions and resulting resident bids and asks. A {\bf process} is a function $(B,A,\tau) \mapsto (B',A',M)$ that takes as input sets of orders $B$, $A$, and an instruction $\tau$ and outputs sets of orders $B'$, $A'$, and a set of transactions $M$.

The input to a process is an order-domain, which represents the resident orders in the system, and an instruction. If this instruction is a delete id instruction, then the process is supposed to delete all resident orders with that id, and the  `effective' order-domain  potentially gets reduced. Otherwise, if the instruction is a buy/sell order, then the `effective' order-domain needs to include that order. We define $\absorb$ that takes an order-domain and an instruction as input and outputs the `effective' order-domain. For an order-domain $(B,A)$ and an instruction $\tau$ we define

$$\absorb(B,A,\tau):=
\begin{cases}
(\{\beta \in B \mid \id(\beta)\neq id\},\{\alpha \in A  \mid \id(\alpha)\neq id\}) &\text{ if } \tau=\text{Del } id\\ (B\cup\{\beta\},A) &\text{ if } \tau=\text{Buy } \beta\\
(B,A\cup\{\alpha\})&\text{ if } \tau=\text{Sell } \alpha.
\end{cases}$$

Observe that the following propositions follow immediately from the definition of $\absorb$.

\begin{proposition}\label{prop:absorb1}
If $\tau$ is an instruction and $(B,A)$ is an order-domain such that the timestamps of the orders in $B\cup A$ are all distinct and different from the timestamp  appearing in $\tau$ and if $(B',A')=\absorb(B,A,\tau)$, then the timestamps of the orders in $B'\cup A'$ are all distinct.
\end{proposition}

\begin{proposition}\label{prop:absorb2}
If $(B,A)$ is an order-domain such that the ids of the orders in $B\cup A$ are all distinct, then for all $\del$ instructions $\tau$, if $(B',A')=\absorb(B,A,\tau)$, then the ids of the orders in $B'\cup A'$ are all distinct.
\end{proposition}

\begin{proposition}\label{prop:absorb3}
If $\tau$ is an instruction and $(B,A)$ is an order-domain such that the ids of the orders in $B\cup A$ are all distinct and different from the id appearing in $\tau$ and if $(B',A')=\absorb(B,A,\tau)$, then the ids of the orders in $B'\cup A'$ are all distinct.
\end{proposition}

To a process, we will usually feed inputs that satisfy certain properties, and such inputs we refer to as legal-inputs and are defined as follows. We say an order-domain $(B,A)$ and an instruction $\tau$ forms a {\bf legal-input} if $B$ and $A$ are not matchable and  $\tau$ is such that $(B',A')=\absorb(B,A,\tau)$ is an admissible order-domain.

\begin{figure}
  \begin{minipage}[c]{0.50\textwidth}
    \includegraphics[scale=.72,trim=0cm 0 1.6cm 0cm]{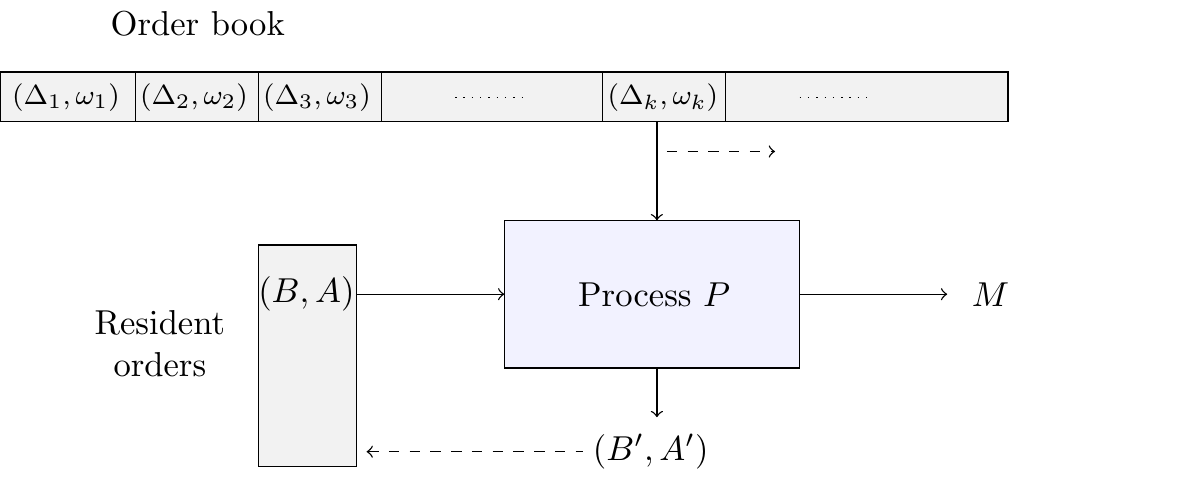}
  \end{minipage}\hfill
  \begin{minipage}[c]{0.43\textwidth}
\caption*{Illustration of a Process: at time $k$ the process $P$ takes as input resident orders $(B,A)$ and an instruction $(\Delta_k, \omega_k)$ and outputs a matching $M$ and a pair of sets of orders $(B',A')$ that will act as resident orders for the next time step. $\iter$ takes as input a process $P$, an order book $I$, and a time $k$ and outputs what $P$ would output on $I$ at time $k$. Thus, for the example in the figure, we have $\iter(P, I, k) = (B', A', M)$.}
  \end{minipage}
\end{figure}

Finally, we define $\iter$. Given a process $P$, an order book $\mathcal{I}$ and a natural number $k$, we define $\iter(P,\mathcal I, k)$ to be the output of $P$ at time $k$ when it is iteratively run on the order book $\mathcal I$. When $k>\mathsf{length}(\mathcal I)$, $\iter(P,\mathcal I, k)$ returns $(\emptyset, \emptyset, \emptyset)$. Otherwise, $\iter(P,\mathcal I, k)$ can be computed recursively as per the following algorithm.

\begin{algorithm}[H]
 \caption{Iteratively running a process on an order book}\label{process}
 \includegraphics[scale =.75]{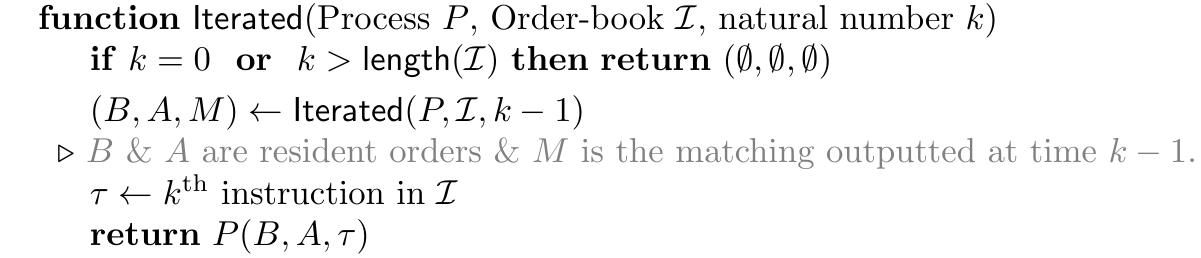}
\end{algorithm}

\section{Three natural properties of a process}\label{sec:properties} 

Having set up the definitions, we will now state the three natural properties formally.

We say a process $P$ satisfies \texttt{positive bid-ask spread}, \texttt{price-} \texttt{time priority}, and \texttt{conservation} if for all order-domains $(B,A)$ and an instruction $\tau$ such that $(B,A)$ and $\tau$ forms a legal-input,
 $P(B,A,\tau)=(\hat B,\hat A, M)$ and $(B',A')=\absorb(B,A,\tau)$ implies the following three conditions.
 
\begin{enumerate}
    \item {\bf \texttt{Positive Bid-Ask Spread}}: $\hat{B}$ and $\hat{A}$ are not matchable.
    
    \item {\bf \texttt{Price-Time Priority}}: If a less competitive order $\omega$ gets traded in $M$, then all orders that are more competitive than $\omega$ must be fully traded in $M$. Formally,
    \begin{align*}
    \text{a. } \ \forall &a,a'\in A', \ a \succ a' \text{ and } \id(a') \in \idsa(M) 
    \implies \Q(M,id(a)) = \q(a) \\
    \text{b. } \ \forall &b,b'\in B', \ b \succ b' \text{ and } \id(b') \in \idsb(M) 
    \implies \Q(M,id(b)) = \q(b).
    \end{align*}
    
    \item {\bf \texttt{Conservation}}: $P$ does not lose or add orders arbitrarily. For this, we have the following technical conditions.
\begin{align*} 
 &\text{a. } \ M \text{ is a matching over the order-domain } (B',A') \\
&\text{b. } \ \hat{B} = B' - \Bids(M,B') \quad
\text{c. } \  \hat{A} = A' - \Asks(M,A').
\end{align*}
\end{enumerate}

The above definition appears in our Coq formalization as follows.

\includegraphics{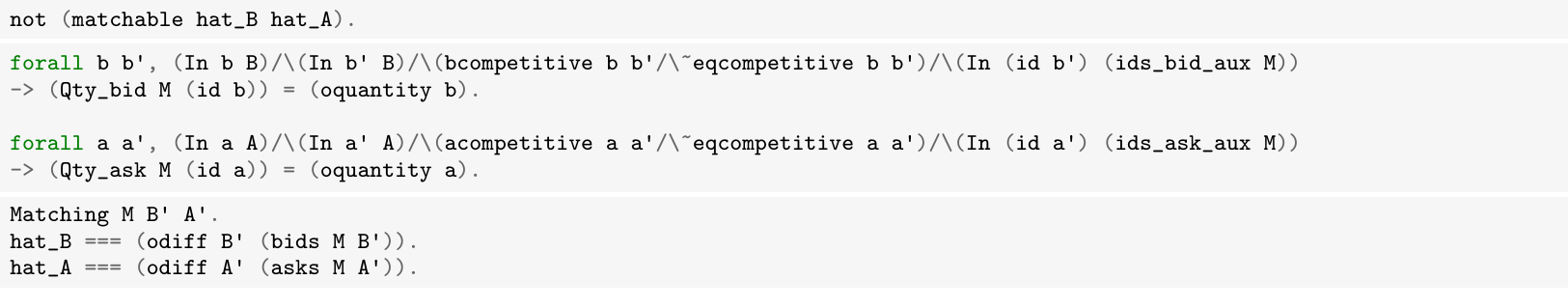}


\section{Maximum matching}\label{sec:maximum}
In this section, we prove the maximum matching lemma.
\begin{lemma}[Maximum matching]
\label{lem:max}
Let $P$ be a process that satisfies \texttt{positive} \texttt{bid-ask} \texttt{spread} and \texttt{conservation}.
For all order-domain and instruction pairs $((B,A),\tau)$ that form legal-inputs, if $P(B,A,\tau)=(\hat B, \hat A, M)$, then for all matchings $M'$ over $\absorb(B,A,\tau)$, $\vol(M)\geq \vol(M')$.
\end{lemma}
\begin{proof}[Proof of Lemma~\ref{lem:max}]
Let us fix a process $P$, an order-domain $(B, A)$, and an instruction $\tau$ as in the lemma statement. Let  $(\hat B, \hat A, M)$ $=$ $P(B,A,\tau)$ and $(B',A')=\absorb(B,A,\tau)$. Let $M'$ be an arbitrary matching over $(B',A')$. We need to show $\vol(M)\geq \vol(M')$.

Now we have the following three cases depending on $\tau$.

\noindent{\bf Case: $\tau$ is $\del \ id$.} In this case $B' = \{b\in B| \id(b)\neq id\}\subseteq B$ and $A'=\{a\in A|\id(a)\neq id\}\subseteq A$ as $(B',A')=\absorb(B,A,\tau)$.

Since $(B,A)$ and $\tau$ forms a legal input, $(B,A)$ is not matchable. Consequently, $(B',A')=\absorb(B,A,\tau)$ is also not matchable as $B'\subseteq B$ and $A'\subseteq A$. Thus, no valid transaction over $(B',A')$ exists. Since the matching $M'$ is a set of valid transactions over $(B',A')$,  $M'=\emptyset$. Thus, $\vol(M')=0$, and we are trivially done.

\noindent{\bf Case: $\tau$ is  $\buy \ \beta$.}
In this case $B'=B\cup\{\beta\}$ and $A'=A$ as
$(B',A')=\absorb(B,A,\tau)$.

If $(B',A')$ is not matchable, then as in the previous case, $\vol(M')$ $=$ $0$, and we are trivially done.
Thus, we may assume $(B',A')$ is matchable. Now since $(B,A)$ is not matchable, only $\beta \in B'$ is tradable with some ask in $A'$. Thus, for all valid transactions $t$ over $(B',A')$, $\idb(t)=\id(\beta)$. 

We assume for contradiction: $\vol(M) < \vol(M')$. We will show that there exists a bid $\hat b\in\hat B$ and an ask $\hat a\in \hat A$ that remain resident and are tradable, contradicting \texttt{positive bid-ask spread}.

Since $M$ and $M'$ are matchings over $(B',A')$ consisting of only valid transactions $t$ such that $\idb(t)=\id(\beta)$ and $\vol(M) < \vol(M')$, there must exist an ask $a\in A'$ which is tradable with $\beta$, such that its trade quantity in $M$ is strictly less than its trade quantity in $M'$, i.e.,
$$ \Q_{\text{ask}}(M,\id(a)) < \Q_{\text{ask}}(M',\id(a))\leq\q(a), $$
$\text{ and } \beta \text{  and } a \text{ are tradable.}$ Thus, from \texttt{conservation} part c, a part of $a$ will remain resident in $\hat A$. In particular, there exists $\hat a \in \hat A$ such that $\pr(\hat a)=\pr(a)$. Now, since $\beta$ and $a$ are tradable, so are $\beta$ and $\hat a$.

 Again since $M'$ is a matching over $(B',A')$ consisting of only valid transactions $t$ such that $\idb(t)=\id(\beta)$, $\vol(M') \leq \q(\beta)$. Thus, from our assumption, $\vol(M) < \vol(M')\leq \q(\beta)$ which implies, from \texttt{conservation} part b, some part of $\beta$ remains untraded in $M$. In particular,
there exists $\hat \beta \in \hat B$ such that $\pr(\hat \beta)=\pr(\beta)$. Now since $\beta$ and $\hat a$ are tradable, $\hat \beta$ and $\hat a$ are tradable, which contradicts \texttt{positive bid-ask spread}, completing the proof.

\noindent{\bf Case: $\tau$ is  $\sell \ \alpha$.}
The proof in this case is symmetric to the above case.
\end{proof}

\section{Local Uniqueness}
In this section, we prove the local uniqueness theorem.
\begin{theorem}
\label{thm:processUniqueness}
Let $P_1$ and $P_2$ be processes that satisfy \texttt{price-time priority}, \texttt{positive bid-ask spread}, and \texttt{conservation}.
For all order-domain instruction pairs $((B,A),(\Delta,\omega))$ that form legal-inputs, if for each $i\in\{1,2\}$, $P_i(B,A,(\Delta, \omega))=(\hat B_i, \hat A_i, M_i)$, then \\
$(\hat {B_1},\hat{A_1},\mathcal C(M_1))=(\hat{B_2},\hat{A_2},\mathcal C(M_2))$. Furthermore, the following statements hold for each $i$.

\noindent    (1) $\ \hat{B_i} \text{ and } \hat{A_i} \text{ are not matchable. }$ \\
    (2) $\text{ The timestamps of orders in } \hat{B_i} \text{ and }\hat{A_i} \text{ are distinct} \text{ and form a} \text{ subset of the set of}$ \\
    $\text{ timestamps of the orders in } B\cup A\cup \{\ts(\omega)\}.$\\
    (3) $\text{ The ids of orders in } \hat{B_i} \text{ and }\hat{A_i} \text{ are distinct and form a subset} 
    \text{ of }  \ids(B\cup A)\cup\{\id(\omega)\}.$ \\
    (4) $ \ \Delta = \del \implies \id(\omega) \notin \ids(\hat{B_i}\cup \hat{A_i}).$
\end{theorem}

\begin{proof}
Fix processes $P_1$ and $P_2$ that satisfy the three properties. Fix an order-domain $(B,A)$ and an instruction $(\Delta, \omega)$ which forms a legal-input. Let for each $i\in\{1,2\}$ $(\hat B_i, \hat A_i, M_i) = P_i(B,A,(\Delta,\omega))$. Let $(B',A')=\absorb(B,A,(\Delta,\omega))$.

We will show $(\hat B_1, \hat A_1,\mathcal C(M_1))=(\hat B_2,\hat A_2,\mathcal C(M_2))$. 

Proofs of (1)-(4) above are relatively straightforward, and we omit detailed proofs:
(1) is precisely \texttt{positive bid-ask spread}. (2) and (3) follow from the fact that timestamps and ids of the resident orders arise from the timestamps of the orders in the order-domain $(B',A')$, which is admissible. (4) follows from the fact that if $\Delta = \del$, then $\id(w)\notin \ids(A'\cup B')$.

To show $(B_1,A_1,\mathcal C(M_1))=(B_2,A_2,\mathcal C(M_2))$, we will do a case analysis based on $\Delta$.

First note that both $M_1$ and $M_2$ are maximum volume matching between $B'$ and $A'$ from Lemma~\ref{lem:max}. 

\noindent {\bf Case 1: $\Delta$ is $\del$.} In this case $B'\subseteq B$ and $A'\subseteq A$ as $(B',A')=\absorb(B,A,(\Delta,\omega))$. Now, since $(B,A)$ is not matchable, $(B',A')$ is also not matchable, implying $M_1=M_2=\emptyset$ as no valid transactions exist over the order-domain $(B',A')$. From \texttt{conservation} parts b and c, we have $\hat B_i=B'$ and $\hat A_i = A'$ for each $i\in\{1,2\}$, and we are done.

\noindent {\bf Case 2: $\Delta$ is $\buy$.} In this case $B'= B\cup \{\omega\}$ and $A'=A$ as $(B',A')=\absorb(B,A,(\Delta,\omega))$. Also since $(B,A)$ is not matchable, then every valid transaction $t$ over the order-domain $(B',A')$ is such that $\idb(t)=\id(\omega)$.

If $\omega$ is not tradable with any ask in $A'$, then we are done like in the previous case as there are no possible valid transactions over $(B',A')$, implying $M_1=M_2=\emptyset$. And as before $B_i=B'$ and $A_i=A'$ for each $i\in\{1,2\}$ follows from \texttt{conservation}, and we are done.

Thus, we may assume $\omega$ is tradable with some ask in $A'$. We will first show that for each ask $a\in A$, the total traded volume of $a$ in $M_1$ is equal to the total traded volume of $a$ in $M_2$. Assume for the sake of contradiction, there exists an ask $a$ that has more traded quantity in $M_1$ than in $M_2$. This implies $a$ is partially traded in $M_2$. Since trade volumes of both $M_1$ and $M_2$ are equal (as they are both maximum), there must be an ask $a'$ such that total traded volume of $a'$ in $M_2$ is more than that in $M_1$. In particular, this implies that $a'$ has at least one transaction in $M_2$. 

Now either $a \succ a'$ or $a' \succ a$. First assume $a \succ a'$. The fact that a more competitive order $a$ is partially traded in $M_2$ while a less competitive order $a'$ is being traded contradicts \texttt{price-time priority}. Next assume $a' \succ a$, the fact that in $M_1$, $a'$ is partially traded and $a$ is also traded, leads to a similar contradiction. Thus, for all $a\in A'$, $\Q_\text{ask}(M_1,\id(a))=\Q_\text{ask}(M_2,\id(a))$.

Now recall in each valid transaction $t$ over the order-domain $(B',A')$ $\idb(t)=\id(\omega)$. As $M_i$ consists of only such transactions $t$, $\mathcal C (M_i)$ consists of transactions of the form $(\id(\omega),\id(a),q)$ where $q=\Q_\text{ask}(M_i, \id(a))$ and $a\in A'$. Now since  for all $a\in A'$, $\Q_\text{ask}(M_1,$ $\id(a))=\Q_\text{ask}(M_2,\id(a))$, a transaction is in $\mathcal C (M_1)$ iff it is in $\mathcal C(M_2)$, proving $\mathcal C (M_1)=\mathcal C(M_2)$.

Now we prove $\hat{A_1}=\hat{A_2}$. Fix an $i\in\{1,2\}$. It follows from \texttt{conservation} part c that each $a\in A'$ which is not fully traded in $M_i$ by $P_i$ will have a corresponding $\hat{a_i} \in \hat A_i$ such that $\q(\hat{a_i})$ $= \q(a) - \Q_\text{ask}(M_i, \id(a))$ and $(\id(\hat a),\ts(\hat a),\pr(\hat a)) =(\id(a),\ts(a),\pr(a))$. Furthermore, no element exists in $\hat A_i$ whose $\id$, $\ts$, and $\pr$ do not match with the respective attributes of some element in $A'$. Now, we proved earlier that for all $a\in A'$, $\Q_\text{ask}(M_1,\id(a))=\Q_\text{ask}(M_2,\id(a))$. Combining these facts, we get $\hat{A_1}= \hat{A_2}$.

$\hat{B_1} = \hat{B_2}$ follows from \texttt{conservation} part b: Observe that if $\omega$ gets fully traded in $M_1$ and $M_2$, $\hat B_1 = \hat B_2 = B$, as all elements of $B$ remain resident from \texttt{conservation}. If $\omega$ gets partially traded  in $M_1$ and $M_2$, then $\hat B_1$ $=$ $\hat B_2 = B \cup \{\hat\omega\}$, where $(\id(\hat\omega), \ts(\hat\omega), \pr(\hat\omega))$ $= (\id(\omega),  \ts(\omega), \pr(\omega))$ and $q(\hat\omega)= \q(\omega)-\vol(M_i)$ (for each $i\in\{1,2\})$.

\noindent {\bf Case 3: $\Delta$ is $\sell$.} Here, the proof is symmetric to the proof in Case 2.
\end{proof}


\section{Global uniqueness}

In this section, we prove the global uniqueness theorem.

\begin{theorem}\label{thm:globalUniquness}
Let $P_1$ and $P_2$ be processes that satisfy \texttt{positive bid-ask spread}, \texttt{price-time priority}, and \texttt{conservation}. Then, for all structured order books $\mathcal I$ and natural numbers $k$ if $\iter(P_1, \mathcal I, k)$ $=$ $(B_1,A_1,M_1)$ and $\iter(P_2, \mathcal I, k)$ $=$ $(B_2,A_2,M_2)$, then $(B_1,A_1,\mathcal C(M_1))$ $=$ $(B_2,A_2,\mathcal C(M_2))$.
\end{theorem}

In our Coq formalization, this theorem appears as follows.

\includegraphics{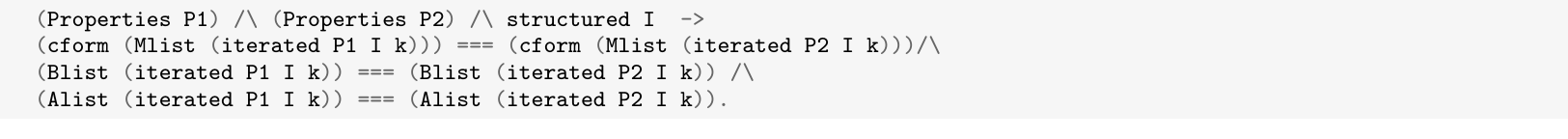}

To prove the above theorem we lift the local uniqueness theorem, Theorem~\ref{thm:processUniqueness}, using induction. We achieve this by strengthening the theorem statement that  provides us with a strong enough induction hypothesis for establishing the pre-conditions needed to apply Theorem~\ref{thm:processUniqueness}.
\begin{proof}[Proof of Theorem~\ref{thm:globalUniquness}.]
Fix processes $P_1$ and $P_2$ and a structured order book $\mathcal I = $ 
$ [(\Delta_0,\omega_0),\cdots$ $,\Delta_n,\omega_n)]$, and let $\ip_i(k)$ denote $\iter(P_i,\mathcal I, k)$ for $i\in\{1,2\}$.

Notice for $k> \mathsf{length}(\mathcal I)$ $=$ $n+1$, $\ip_1(k)=\ip_2(k)$ $=$ $(\emptyset,\emptyset,\emptyset)$ and we are trivially done. Now we will show the theorem statement holds for $k\leq n+1$. More generally, we prove the following statement by induction on $k$.

For all $k\leq n+1$, if  $\ip_1(k) = (B_1^k, A_1^k, M_1^k)$ and $\ip_2(k)$ $=$ $(B_2^k, A_2^k, M_2^k)$, then the following properties hold.

    (1) $(B_1^k,A_1^k, \mathcal C(M_i^k))= (B_2^k,A_2^k,\mathcal C(M_i^k)).$
    
    (2) $B_1^k$  and $A_1^k$  are not matchable. 
    
    (3)   The timestamps of orders in  $B_1^k\cup A_1^k$ \ are distinct and form a subset of the set of timestamps of  $\{\omega_0,\cdots,\omega_{k-1}\}$.
    
    (4) Orders in $B_1^k \cup A_1^k$  have distinct ids and they belong
     to 
    $\ids(\{\omega_0,\cdots,\omega_{k-1}\})$.
    
    (5)  $k\geq 1$ and $\Delta_{k-1} = \del \implies \id(\omega_{k-1}) \notin \ids(B_1^k\cup A_1^k)$.

Observe that proving (1) above would complete the proof of Theorem~\ref{thm:globalUniquness}.

Base case: When $k=0$, $\ip_1(0)=\ip_2(0)=(\emptyset,\emptyset,\emptyset)$, and properties (1)-(5) trivially hold.
 
 We now assume that the above statement (properties (1)-(5)) holds for  $k = t$ and we will show that the above statement holds for $k=t+1$. Notice for each $i\in\{1,2\}$, $\ip_i(t+1)$ first computes $\ip_i(t)$ recursively to obtain $(B_i^{t},A_i^{t})$ and then returns $P_i(B_i^{t},A_i^{t},(\Delta_{t}, \omega_{t}))=(B_i^{t+1},A_i^{t+1},M_i^{t+1})$. Our proof proceeds in two parts. First, using the induction hypothesis we prove that $(B_i^{t},A_i^{t},(\Delta_{t}, \omega_{t}))$ forms a legal input, which is the pre-condition to apply Theorem~\ref{thm:processUniqueness}. Next, we invoke Theorem~\ref{thm:processUniqueness} and establish that properties (1)-(5) hold for $k=t+1$.
 
 {\bf First part:} $(B_i^{t},A_i^{t},(\Delta_{t}, \omega_{t}))$ forms a legal input. To show this, observe from the definition of legal input, we need to show (i) $B_i^t$ and $A_i^t$ are not matchable and (ii) $\absorb(B_i^t,$ $A_i^t,(\Delta_{t},\omega_{t}))$ forms an admissible order-domain. (i) is immediate from the induction hypothesis (property (2)). To show (ii) we need to show that if $\absorb(B_i^t,A_i^t,(\Delta_{t},\omega_{t}))=(B',A')$, then the timestamps of the orders in $B'\cup A'$ are all distinct and the ids of the orders in $B'\cup A'$ are all distinct. The timestamps of the orders in $B'\cup A'$ are distinct follows from Proposition~\ref{prop:absorb1}, since from the induction hypothesis (property(3)) it follows that the orders in $B_1^t$ and $A_1^t$ have distinct timestamps from the set $\{\ts(\omega_0),\cdots,\ts(\omega_{t-1})\}$, and they are strictly smaller than $\ts(\omega_{t})$ as $\mathcal I$ is structured. Below we show that the ids of $B'\cup A'$ are distinct by considering two cases: $\Delta_t=\del$ and  $\Delta_t\neq\del$.
\begin{itemize}
    \item Case: $\Delta_t = \del$. We are immediately done from Proposition~\ref{prop:absorb2}, since orders in $B_1^t\cup A_1^t$ have distinct ids from the induction hypothesis (property (4)).
    \item Case: $\Delta_t\neq\del$. Since $\Delta_t\neq\del$ and $\mathcal I$ is structured, either (a) $\id(\omega_{t}) \notin \ids\{\omega_0,\cdots, \omega_{t-1}\}$ or
   (b) $\id(\omega_{t}) = \id(\omega_{t-1})$ and $\Delta_{t-1}=\del$. In either case we claim that $\id(\omega_{t})\notin \ids(A_1^t \cup B_1^t)$. From the induction hypothesis (property (4)), we know orders in $B_1^t\cup A_1^t$ have distinct ids from the set $\ids(\{\omega_0,\cdots,\omega_{t-1}\})$. So in case (a), it immediately follows that $\id(\omega_{t})\notin \ids(A_1^t \cup B_1^t)$. In case (b), since $\Delta_t=\del$, from induction hypothesis (property (5)), we have that $\id(\omega_{t-1})\notin \ids(B_1^t\cup A_1^t)$. In case (b) we also have $\id(\omega_{t})=\id(\omega_{t-1})$. Thus, combining these two facts, we have $\id(\omega_{t})\notin\ids(B_1^t\cup A_1^t)$. Now using the claim and the fact that ids of orders in $B_1^t\cup A_1^t$ are all distinct, we invoke Proposition~\ref{prop:absorb3} to get that ids of orders in $B'\cup A'$ are all distinct.
\end{itemize} 

{\bf Second part:} properties (1)-(5) hold for $k=t+1$.
(1) holds: From the induction hypothesis (property (1)), we have $(B_1^k,A_1^k)=(B_2^k,A_2^k)$. Now since $(B_1^k,A_1^k,(\Delta_t,\omega_t))$ forms a legal input, we are immediately done by invoking Theorem~\ref{thm:processUniqueness}.
(2) and (5) also immediately follow by invoking Theorem~\ref{thm:processUniqueness}.
(4) holds: Theorem~\ref{thm:processUniqueness} implies that the ids of orders in $B_1^{t+1}\cup A_1^{t+1}$ are all distinct and form a subset of $(\ids(B_1^t\cup A_1^t)\cup \{\id(\omega_t)\})$. Now from induction hypothesis (property (4)), we have  $\ids(B_1^t\cup A_1^t)\subseteq\ids(\{\omega_0,\cdots,\omega_{t-1}\})$. Thus, ids of the orders in $B_1^{t+1}\cup A_1^{t+1}$ forms a subset of $\ids(\{\omega_0,\cdots\omega_t\})$, and we are done.
Using an almost identical argument, we can show that (3) holds.\end{proof}


\section{Verified Algorithm}\label{sec:verified}

Here we introduce a natural algorithm for continuous double auctions. 

 \begin{algorithm}
 \caption{Process for continuous market}\label{process_alg}
\includegraphics[scale =.75]{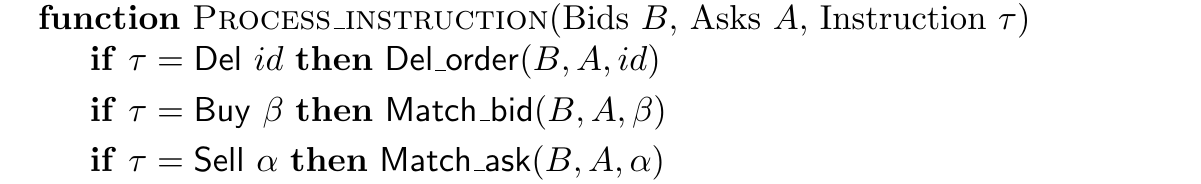}
 \end{algorithm}

In the Coq formalization of Process\_instruction, we sort the list of asks and bids by their competitiveness before calling a subroutine; as a result, the most competitive bid and ask are on top of their respective lists.

 \begin{algorithm}
 \caption{Matching an ask}\label{match}
\includegraphics[scale =.75]{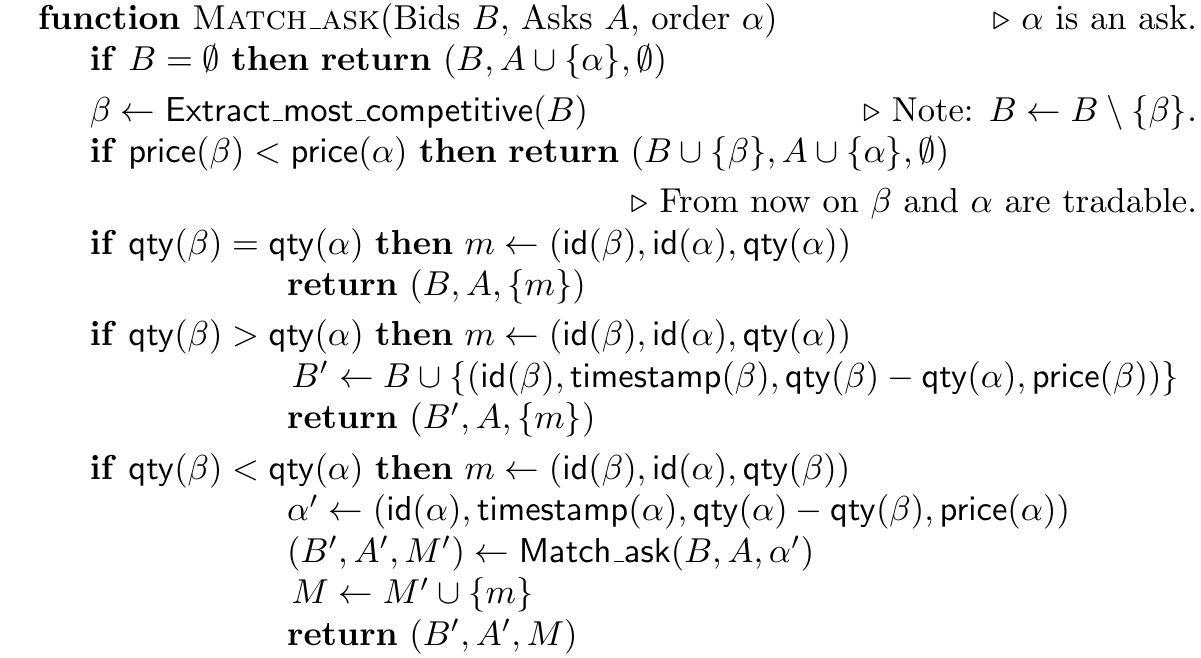}
 \end{algorithm}

The $\matchbid$ subroutine is symmetric to the $\matchask$ subroutine and we do not present it explicitly here.

 \begin{algorithm}
 \caption{Deleting an order}\label{setp}
\includegraphics[scale =.75]{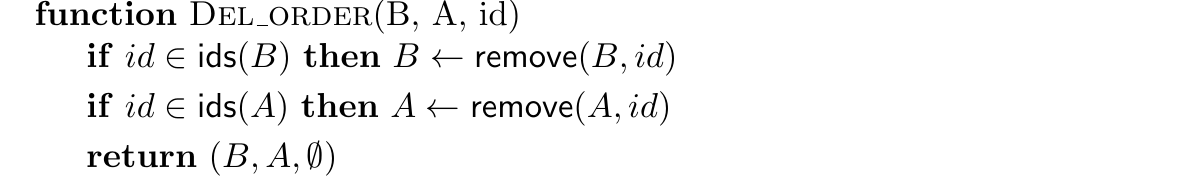}
 \end{algorithm}


Next, we show that the above algorithm  satisfies the three natural properties.
\begin{theorem}\label{thm:Process_instructionSatisfiesProperties}
Process\_instruction satisfies \texttt{positive bid-ask spread}, \texttt{price-time} \\ \texttt{priority}, and \texttt{conservation}.
\end{theorem}
 
 In our Coq formalization, the above theorem statement appears as follows.

\includegraphics{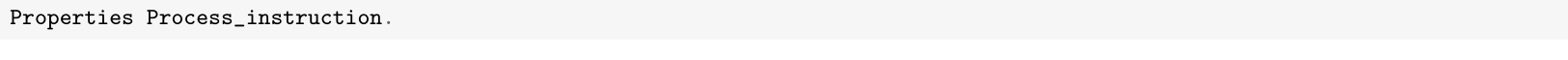}


 The proof of this result is outlined in Appendix~\ref{sec:algoCorrectness}.


\section{Application: checker}
\label{sec:application}

In this section, we discuss how we use our verified algorithm and the global uniqueness theorem to design an automated checker that detects errors in exchange algorithms that run continuous double auctions from their trade logs. We implement such a checker and run it on trade logs from an exchange. We include, as part of the supplementary materials accompanying this paper, the OCaml source code of our checker, trade logs for two example stocks (created from real stocks from a stock exchange after appropriate preprocessing and masking), and a shell script that compiles our code and runs it on the trade logs of the two stocks.

As discussed in the paper, exchanges for each traded product maintain an order book that contains all the incoming orders and a corresponding trade book  that contains all the transactions that are generated as trade logs. These trade logs are accessible to the regulators. We show a market regulator can use these trade logs to automatically determine whether an exchange is complying with the three natural requirements of \texttt{price-time priority}, \texttt{positive bid-ask spread}, and \texttt{conservation} while generating transactions. Recall, as implied by our global uniqueness theorem, these three properties completely specify continuous double auctions.

\subsection{Algorithm for checker}
Given an order book and the corresponding trade book, the checker first runs the verified matching algorithm on the order book to generate `verified' matchings. Then, the canonical forms of the matchings in the trade book with the verified matchings are compared, one time-step at a time. If for any time-step, the canonical forms of the matchings do not match, the checker outputs a mismatch detected message along with the corresponding matchings. Otherwise, if all matchings match, the checker terminates without a mismatch message.

For a given order book and trade book, if the above checker finds a mismatch, then we can conclude, from the global uniqueness theorem, that the exchange algorithm violates at least one of the three properties. Otherwise, at least for the instance at hand, the exchange algorithm output is as good as that of a verified algorithm.

In our implementation of the checker, the verified matching algorithm and the canonical form functions are directly extracted from our formalization using Coq's code extraction feature and used as subroutines.

\subsection{Preprocessing an order book}
Real exchanges implement more complex instruction types than our three primitives: buy/sell/delete. Here we briefly discuss some of these instruction types and how to convert them into primitive types in a preprocessing step.

\begin{itemize}
    \item{Updates.} A buyer or seller might want to update the quantity or price of his order using an update order instruction. For our test exchange, the rule is if the quantity of the order decreases, then the timestamp of the original order is retained. Otherwise, if the quantity is increased or the price changes the new order carries the timestamp of the update instruction. Our test exchange as part of the update instruction also maintains the information of the quantity of the order that remains untraded. Updates can be easily implemented by replacing it with a delete instruction followed by a new buy/sell instruction with the updated attributes.
    \item{Market orders.} Market orders are orders that do not specify a limit price and are ready to be traded at any price. We can implement such an order quite easily by keeping the price $0$ for a sell order and $\infty$ (in our implementation we use the maximum number supported by the system) for a buy order.
    \item{Immediate or cancel orders (IOCs).} An IOC is an order that needs to be immediately removed from the system after it is processed, i.e., it should never become a resident order. Such orders can be implemented by replacing them with a buy/sell instruction followed by a delete instruction.
    \item{Stop-loss orders.} Stop-loss orders are orders that are triggered when certain events happen, like when the price of a transaction goes below a threshold. For our test exchange, the timestamp of the triggering event is provided. Consequently, these orders can be treated as normal orders when inserted in the order book at a position corresponding to the timestamp of the triggering event.
\end{itemize}

Note that the preprocessing step can at the most double the number of instructions in the order book. Apart from these standard orders discussed above, certain exchanges allow for more complicated orders which include `iceberg orders'. Iceberg orders are rare and are more complicated to preprocess where the priority of the orders depends on factors that cannot be fully determined by just price and time. Although we can implement such orders through some preprocessing hacks, ideally the matching algorithm should be enriched to handle such orders.

\subsection{Demonstration}
We include, as part of the supplementary materials~\cite{suppl} accompanying this paper, the OCaml source code of our checker, trade logs for two example stocks (created from real data after appropriate preprocessing and masking), and a shell script that compiles our code and runs it on the trade logs of the two stocks. The checker takes under a second to process both the order books. For the first stock, the order book has about $16000$ instructions; the checker detects that there is no mismatch. For the second stock, the order book has about $15000$ instructions; the checker detects a mismatch. On inspection, we found that a delete instruction appears in the order book before the corresponding order was placed causing the mismatch. 

We believe that it would be difficult to detect such rare anomalies without the systematic application of formal methods; in absence of formal specifications, uniqueness theorems, and verified algorithms, such anomalies can be easily brushed aside as complex and probably unavoidable program behavior.

\section{Conclusions}
In this work, we formally introduced continuous double auctions and specified them in terms of three natural and necessary properties. Then, we showed that a natural algorithm possesses these properties. Based on our work, we were able to produce a checker that automatically detects errors in existing exchange algorithms by going over their trade logs.

There are many opportunities in the field of certified auctions. For continuous double auctions alone, one could look at various sophisticated priority rules. For example, what would be the specifications for parity/priority matchings (as used by NYSE)? What would be the specifications for price-time priority matchings in the presence of iceberg orders? What would be the specifications for decentralized exchanges (for cryptocurrencies)?

With this work and previous works in the field, a new paradigm for formalizing various auction mechanisms is emerging, which we believe has useful policy implications for regulators. With the formal methods technology, it is now possible to formally verify systems and require auctioneers to implement them, and also check if the implementation is correct by checking the logs of the auctions that were conducted.

\bibliography{market}

\begin{thebibliography}{10}

\bibitem{suppl}
Supplementary materials: Coq formalization and demonstration, 2022.
\newblock URL: \url{https://github.com/suneel-sarswat/cda}.

\bibitem{clf}
Iliano Cervesato, Sharjeel Khan, Giselle Reis, and Dragisa \v{Z}uni\'c.
\newblock Formalization of automated trading systems in a concurrent linear
  framework.
\newblock In {\em Linearity-TLLA@FLoC}, volume 292 of {\em EPTCS}, pages 1--14,
  2018.
\newblock URL: \url{http://arxiv.org/abs/1904.06159}.

\bibitem{xetra}
{Frankfurt Stock Exchange}.
\newblock {Market Model for the Trading Venue Xetra}.
\newblock
  \url{https://www.xetra.com/resource/blob/1963528/ceda0de49d5b8db5e33cdb17375e4cc9/data/T7_R.8.1_Market_Model-en.pdf},
  Sep 22, 2021.

\bibitem{D93}
Daniel Friedman.
\newblock The double auction market institution: A survey.
\newblock {\em The double auction market: Institutions, theories, and
  evidence}, 14:3--25, 1993.

\bibitem{harris2003trading}
Larry Harris.
\newblock {\em Trading and exchanges: Market microstructure for practitioners}.
\newblock OUP USA, 2003.

\bibitem{KP18}
Cezary Kaliszyk and Julian Parsert.
\newblock Formal microeconomic foundations and the first welfare theorem.
\newblock In June Andronick and Amy~P. Felty, editors, {\em Proceedings of the
  7th {ACM} {SIGPLAN} International Conference on Certified Programs and
  Proofs, {CPP} 2018, Los Angeles, CA, USA, January 8-9, 2018}, pages 91--101.
  {ACM}, 2018.
\newblock \href {https://doi.org/10.1145/3167100} {\path{doi:10.1145/3167100}}.

\bibitem{libra}
Vasilios Mavroudis and Hayden Melton.
\newblock Libra: Fair order-matching for electronic financial exchanges.
\newblock In {\em Proceedings of the 1st ACM Conference on Advances in
  Financial Technologies}, pages 156--168, 2019.

\bibitem{RSS21}
Raja Natarajan, Suneel Sarswat, and Abhishek~Kr Singh.
\newblock Verified double sided auctions for financial markets.
\newblock In Liron Cohen and Cezary Kaliszyk, editors, {\em 12th International
  Conference on Interactive Theorem Proving, {ITP} 2021, June 29 to July 1,
  2021, Rome, Italy (Virtual Conference)}, volume 193 of {\em LIPIcs}, pages
  28:1--28:18. Schloss Dagstuhl - Leibniz-Zentrum f{\"{u}}r Informatik, 2021.
\newblock \href {https://doi.org/10.4230/LIPIcs.ITP.2021.28}
  {\path{doi:10.4230/LIPIcs.ITP.2021.28}}.

\bibitem{NiuP13}
Jinzhong Niu and Simon Parsons.
\newblock Maximizing matching in double-sided auctions.
\newblock In {\em International conference on Autonomous Agents and Multi-Agent
  Systems, AAMAS '13, Saint Paul, MN, USA, May 6-10, 2013}, pages 1283--1284,
  2013.

\bibitem{imandra}
Grant Passmore, Simon Cruanes, Denis Ignatovich, Dave Aitken, Matt Bray, Elijah
  Kagan, Kostya Kanishev, Ewen Maclean, and Nicola Mometto.
\newblock The imandra automated reasoning system (system description).
\newblock In {\em International Joint Conference on Automated Reasoning}, pages
  464--471. Springer, 2020.

\bibitem{PI17}
Grant~Olney Passmore and Denis Ignatovich.
\newblock Formal verification of financial algorithms.
\newblock In {\em 26th International Conference on Automated Deduction,
  Proceedings}, volume 10395 of {\em Lecture Notes in Computer Science}, pages
  26--41. Springer, 2017.

\bibitem{ptpandprorata}
Tobias Preis.
\newblock Price-time priority and pro rata matching in an order book model of
  financial markets.
\newblock In {\em Econophysics of Order-driven Markets}, pages 65--72.
  Springer, 2011.

\bibitem{SS20}
Suneel Sarswat and Abhishek~Kr Singh.
\newblock Formally verified trades in financial markets.
\newblock In {\em Formal Methods and Software Engineering - 22nd International
  Conference on Formal Engineering Methods, {ICFEM} 2020, Singapore, Singapore,
  March 1-3, 2021, Proceedings}, volume 12531 of {\em Lecture Notes in Computer
  Science}, pages 217--232. Springer, 2020.
\newblock \href {https://doi.org/10.1007/978-3-030-63406-3\_13}
  {\path{doi:10.1007/978-3-030-63406-3\_13}}.

\bibitem{nse}
{Securities Exchange Board of India (SEBI)}.
\newblock {Order in the matter of NSE Colocation}, Apr 30, 2019.
\newblock
  \href{https://www.sebi.gov.in/enforcement/orders/apr-2019/order-in-the-matter-of-nse-colocation_42880.html}{Order
  in the matter of NSE Colocation}.

\bibitem{coq}
The Coq~Development Team.
\newblock The coq reference manual, release 8.12.2, December~11 2020.
\newblock URL:
  \url{https://github.com/coq/coq/releases/download/V8.12.2/coq-8.12.2-reference-manual.pdf}.

\bibitem{ubs}
{U.S. Securities and Exchange Commision (SEC)}.
\newblock {SEC Charges UBS Subsidiary With Disclosure Violations and Other
  Regulatory Failures in Operating Dark Pool}.
\newblock \url{https://www.sec.gov/news/pressrelease/2015-7.html}, July, 2015.

\bibitem{nyse2}
{U.S. Securities and Exchange Commision (SEC)}.
\newblock {NYSE to Pay US Dollar 14 Million Penalty for Multiple Violations}.
\newblock \url{https://www.sec.gov/news/press-release/2018-31}, March 6, 2018.

\bibitem{nyse1}
{U.S. Securities and Exchange Commision (SEC)}.
\newblock {SEC Charges NYSE for Repeated Failures to Operate in Accordance With
  Exchange Rules}.
\newblock \url{https://www.sec.gov/news/press-release/2014-87}, May 1, 2014.

\bibitem{formalpractice}
Jim Woodcock, Peter~Gorm Larsen, Juan Bicarregui, and John Fitzgerald.
\newblock Formal methods: Practice and experience.
\newblock {\em ACM computing surveys (CSUR)}, 41(4):1--36, 2009.

\bibitem{WWW98}
Peter~R. Wurman, William~E. Walsh, and Michael~P. Wellman.
\newblock Flexible double auctions for electronic commerce: theory and
  implementation.
\newblock {\em Decision Support Systems}, 24(1):17--27, 1998.

\bibitem{nuclearformal}
Junbeom Yoo, Taihyo Kim, Sungdeok Cha, Jang-Soo Lee, and Han~Seong Son.
\newblock A formal software requirements specification method for digital
  nuclear plant protection systems.
\newblock {\em Journal of Systems and Software}, 74(1):73--83, 2005.

\bibitem{zhao2010maximal}
Dengji Zhao, Dongmo Zhang, Md~Khan, and Laurent Perrussel.
\newblock Maximal matching for double auction.
\newblock In {\em Australasian Joint Conference on Artificial Intelligence},
  pages 516--525. Springer, 2010.

\end{thebibliography}


\appendix\label{appendix}
\counterwithin{theorem}{section}


\section{Correctness of Process\_instruction} \label{sec:algoCorrectness} 


As Process\_instruction on each input invokes one of $\matchask$, $\matchbid$, and $\deleteorder$, it is enough to show that these subroutines have the desired properties. It is relatively straightforward to show that $\deleteorder$ has the properties and the proof for $\matchbid$ and $\matchask$ are completely symmetric. Thus, in our presentation, we just show that $\matchask$ has the desired properties.

\subsection{Match\_Ask satisfies positive bid-ask spread}

 The $\matchask$ subroutine satisfies \texttt{positive bid-ask spread} follows immediately from the following lemma.

\begin{lemma}\label{lem:Match_AskSatisfiesPositiveBidAskSpread}
Let $(B,A\cup\{\alpha\})$ be an admissible order-domain. If $\matchask(B,A,\alpha)=(\hat B, \hat A, M)$ and $(B,A)$ is not matchable, then $(\hat B,\hat A)$ is not matchable.
\end{lemma}

\begin{proof}
Fix an $A$ and an $\alpha$.  We prove the lemma by induction on $|B|$.

\noindent\underline{Base case}: trivial.

\noindent\underline{Induction hypothesis}:  For all $|B|<k$, If $(B,A\cup\{\alpha\})$ is an admissible domain, $(B,A)$ is not matchable, and $\matchask(B,A,\alpha)=(\hat B, \hat A, M)$, then $(\hat B,\hat A)$ is not matchable.

\noindent\underline{Induction step}: We prove the statement when $|B|=k$.

Let $\beta$ be the most competitive bid in $B$. 
If $\beta$ is not tradable with $\alpha$, then as per $\matchask$, 
$\hat B = B$ and $\hat A = A\cup \{\alpha\}$. Now since $\alpha$ is not tradable with the most competitive bid in $B$, it is not tradable with any bid in $B$. Coupled with the fact that $(B,A)$ is not matchable, we get $(\hat B, \hat A)$ is not matchable.

From now on, we may assume that $\beta$ and $\alpha$ are tradable. We have three cases based on the relationship between $\q(\beta)$ and $\q(\alpha)$.
\begin{itemize}
\item{ If $\q(\beta)=\q(\alpha)$,} then as per $\matchask$, $(\hat B,\hat A)=(B\setminus \{\beta\},A)$, and we are trivially done as $(B,A)$ is not matchable.
\item{If $\q(\beta)>\q(\alpha)$,} then as per $\matchask$, $(\hat B, \hat A) = ((B\setminus\{\beta\})\cup\{\beta'\},A)$ where $\beta$ and $\beta'$ only differ in their quantities. Now since $(B,A)$ is not matchable, $(\hat B,\hat A)$ is not matchable, and we are done. 
\item{If $\q(\beta) < \q(\alpha)$,} then $\matchask$ on $(B,A,\alpha)$ first recursively calls $\matchask$ on $(B\setminus \{\beta\},A,\alpha')$ where $\alpha'$ is obtained from $\alpha$ by reducing its quantity to $\q(\alpha)-\q(\beta)$. Let $(B',A',M')=\matchask (B\setminus\{\beta\},A,\alpha)$. Then, $\matchask$ returns $(\hat B, \hat A, M)= (B',A',M)$ where $M$ is obtained from $M'$ in some way. Now since, the set of bids in the recursive call has cardinality $|B\setminus \{\beta\}|=k-1$, $(B \setminus \{\beta\}, A \cup \{\alpha'\})$ is admissible (as $(B,A\cup\{\alpha\})$ admissible), and $(B\setminus\{\beta\},A)$ is not matchable (as $(B,A)$ not matchable), we can apply the induction hypothesis to conclude $(B',A')$ is not matchable. Now as $(\hat B, \hat A)=(B',A')$, we are done.
\end{itemize}
\end{proof}

\subsection{Match\_Ask satisfies price-time priority}

The subroutine $\matchask$ satisfies \texttt{price-time priority} follows immediately from the following lemma.

\begin{lemma}\label{lem:Match_AskSatisfiesPriceTimePriority}
Let $(B,A\cup\{\alpha\})$ be an admissible order-domain such that $(B,A)$ is not matchable. If $\matchask(B,A,\alpha)=(\hat B, \hat A, M)$, then
\begin{align*}
    \text{(i) }\forall a,a'\in A\cup\{\alpha\}, \ & a \succ a' \text{ and } \id(a') \in \idsa(M)
    \implies  \Q(M,id(a)) = \q(a)  \\
    \text{and (ii) }\forall b,b'\in B, \ & b \succ b' \text{ and } \id(b') \in \idsb(M)      \implies \Q(M,id(b)) = \q(b).
\end{align*}
\end{lemma}

\begin{proof}
We first show (i). If $M=\emptyset$, then \texttt{price-time priority} trivially holds.
Otherwise, since $(B,A)$ is not matchable, only $\alpha$ in $A\cup\{\alpha\}$ is tradable with some bid in $B$. Therefore, $\alpha$ is the most-competitive ask in $A\cup \{ \alpha\}$ and all transactions $t\in M$  are such that $\ida(t)=\id(\alpha)$. Thus no $a'$ exists such that $a\succ a'$ and $\id(a')\in\idsa(M)$. Thus, \texttt{price-time priority} trivially holds.

Next, we prove (ii) by induction on $|B|$. 

\noindent\underline{Base case} is trivial as $B=\emptyset \implies M=\emptyset$. 

\noindent\underline{Induction hypothesis}: We assume that (ii) holds for all $B$ such that $|B|<k$.

\noindent\underline{Induction step}: We now show that (ii) holds for all $B$ such that $|B|=k$.

Fix $b, b'\in B$ such that $b\succ b'$ and $b'$ gets traded in $M$. We need to show that $b$ is fully traded in $M$. 

Let $\beta$ be the most competitive bid in $B$. If $\beta$ is not tradable with $\alpha$, as per $\matchask$, $M=\emptyset$ and we are trivially done. Henceforth we assume $\beta$ is tradable with $\alpha$.

Observe, $\matchask$ is such that it completely trades $\beta$ in its output $M$ (along with possibly other bids from $B$) or only trades a part of $\beta$ but no other bid from $B$. Thus, $\beta$ cannot be $b'$. We thus consider two cases: $b=\beta$ or $b\neq \beta$. 
\begin{itemize}
\item{Case: $b=\beta$.} Here again, we are done from the above observation that either $\beta$ gets completely traded or $\beta$ gets partially traded but no other bid gets traded at all.

\item{Case: $b\neq \beta$.} When $\q(\beta)=\q(\alpha)$ or $\q(\beta)>\q(\alpha)$, as per $\matchask$, $\beta$ is the only bid that gets partially or fully traded in $M$, contradicting the existence of $b$ and $b'$. 

 \hspace{20pt}In the final case, when $\q(\beta) < \q(\alpha)$, $\matchask$ on $(B,A,\alpha)$ recursively calls $\matchask$ on $(B\setminus\{\beta\},A,\alpha')$ where $\alpha'$ is obtained from $\alpha$ by reducing its quantity by $\q(\beta)$.
Let the recursive call return $(B',A',M')=\matchask(B\setminus\{\beta\},A,\alpha')$, then $\matchask$ returns $(\hat B, \hat A, M)$ $=$ $(B',A',M\cup \{m\})$, where $m$ is a transaction between $\alpha$ and $\beta$. Thus, all transactions $t$ in $M$  where $\idb(t)\neq \id(\beta)$ come from $M'$. Now since $\beta\succ b \succ b'$ are distinct, if $b'$ gets traded in $M$, it must get traded in $M'$ and applying the induction hypothesis (as it is easy to verify that |$B\setminus\{\beta\}|=k-1$, $(B\setminus \{\beta\}, A\cup\{\alpha'\})$ admissible and $(B\setminus\{\beta\},A)$ not matchable), we have $b$ gets fully traded in $M'$ and hence in $M$ (as $M= M'\cup\{m\}$) and we are done.
\end{itemize}
\end{proof}

\subsection{Match\_Ask satisfies conservation}

$\matchask$ satisfies \texttt{conservation} follows immediately from the following lemma.

\begin{lemma}\label{lem:Match_AskSatisfiesConservation}
Let $(B,A\cup\{\alpha\})$ be an admissible order-domain. If $\matchask(B,A,\alpha)=(\hat B, \hat A, M)$, then
\begin{align*} 
 &\text{a. } M \text{ is matching over the order-domain } (B,A\cup\{\alpha\}) \\
&\text{b. }\hat{B} = B - \Bids(M,B) \\
&\text{c. } \hat{A} = (A\cup\{\alpha\}) - \Asks(M,A\cup\{\alpha\}).
\end{align*}
\end{lemma}

\begin{proof}

We briefly outline the proof stressing only the important aspects.

\noindent\textbf{Part a.} To show $M$ is a matching over $(B,A\cup\{\alpha\})$, we need to show that 
(i) each transaction in $M$ is between a tradable bid-ask pair where the bid is in $B$ and the ask is in $A\cup\{\alpha\}$, and 
(ii) for each order $\omega$ in $B\cup A\cup\{\alpha\}$, its total transaction quantity in $M$ $\Q(M,\id(\omega))$ is at most its quantity $\q(\omega)$. 

We will show this by induction on $|B|$. The base case is trivial: $|B|=0\implies M=\emptyset$, which is trivially a matching. We assume that (i) and (ii) hold for all $B$ such that $|B|<k$. We now show that (i) and (ii) hold assuming $|B|=k$.

Let $\beta$ be the most competitive bid in $B$. When $\beta$ is not tradable with $\alpha$, as per $\matchask$, $M=\emptyset$ and we are trivially done. Henceforth we assume $\beta$ is tradable with $\alpha$. Now, when $\q(\beta)\geq \q(\alpha)$, $\matchask$ outputs only a single transaction $m$ between $\beta$ and $\alpha$ with transaction quantity $\q(m)=\q(\alpha) \leq \q(\beta)$ and clearly $M=\{m\}$ satisfies (i) and (ii) above.

When $\beta$ is tradable with $\alpha$ and $\q(\beta) < \q(\alpha)$, $\matchask$ creates a transaction $m$ between $\beta$ and $\alpha$ with transaction quantity $\q(\beta)$ and obtains $\alpha'$ from $\alpha$ by reducing its quantity by $\q(\beta)$. It then recursively calls $\matchask(B\setminus \{\beta\},A,\alpha')$ to obtain the matching $M'$ and then outputs the matching $M=M'\cup \{m\}$. 

Note that $(B,A\cup\{\alpha\})$ admissible implies $(B\setminus\{\beta\},A\cup \{\alpha'\})$ is admissible. Furthermore, since $|B\setminus\{\beta\}|=k-1$, we can apply the induction hypothesis to obtain that $M'$ is a matching between $B\setminus \{\beta\}$ and $A\cup\{\alpha'\}$. In particular, this means that $\beta$ does not participate in $M'$ and the total transaction quantity of $\alpha$ in $M'$ is at most $\q(\alpha')=\q(\alpha)-\q(\beta)$. So the total transaction quantity of $\beta$ in $M=M'\cup\{m\}$ is the transaction quantity of $m$, which is $\q(\beta)$, and the total transaction quantity of $\alpha$ in $M=M'\cup\{m\}$ is at most $\q(\alpha')+\q(m)=(\q(\alpha)-\q(\beta)) + \q(\beta) = \q(\alpha)$. With these facts, one can easily verify that (i) and (ii) hold.

\noindent\textbf{Part b.} Here we need to show that $\hat B = B - \Bids(M,B)$. Once again we argue by induction on $|B|$. Let $\beta$ be the most competitive bid in $B$.

In all but the last if statement of $\matchask$ (where $\beta$ and $\alpha$ are tradable and $\q(\beta)<\q(\alpha)$), we have $B$, $\hat B$, and $M$ explicitly and we can directly verify that the property holds.

In the last if condition, $\matchask$ first creates a transaction $m$ between $\beta$ and $\alpha$ with transaction quantity $\q(\beta)$ and obtains $\alpha'$ from $\alpha$ by reducing its quantity to $\q(\alpha)-\q(\beta)$. It then computes $(B',A',M')=\matchask(B\setminus\{\beta\},A,\alpha')$ recursively. Finally it returns $(\hat B, \hat A, M)=(B',A', M' \cup \{m\})$.

For the recursive call, we can apply the induction hypothesis. Thus,
\begin{align*}
&B' = (B \setminus \{\beta\}) - \Bids(M',B\setminus \{\beta\}) \quad\quad\text{(from I.H.)}\\
\implies\quad &\hat B = (B \setminus \{\beta\}) - \Bids(M',B\setminus \{\beta\}) \quad\quad\text{(since $\hat B = B'$)}\\
& \hspace{7.5pt}= B - \Bids(M,B),
\end{align*}
where the last equality follows from the observation that $\beta$ is missing from both $B\setminus\{\beta\}$ and $\Bids(M',B\setminus\{\beta\})$, and $\Bids(M,B)$ contains $\beta$ (as $M=M'\cup\{m\}$ where $m$ is a transaction involving $\beta$ with transaction quantity $\q(\beta)$).

\noindent\textbf{Part c.} Observe that all transactions produced by $\matchask$ have $\alpha$ as its participating ask. 
Thus, one can check that we will be done if we can show that $$\q(\hat A, \id(\alpha)) = \q(\alpha) - \Q(M,\id(\alpha)).$$
To show the above, we again induct on $|B|$.
In all but the last if statement of $\matchask$, we have $\alpha$, $\hat A$, and $M$ explicitly, and we can directly verify that the above equation holds.

In the last if statement, we have that $\q(\beta) <\q(\alpha)$ and $\beta$ and $\alpha$ are tradable, where $\beta$ is the most competitive bid in $B$. In this case, $\matchask$ creates a transaction $m$ between $\beta$ and $\alpha$ with transaction quantity $\q(\beta)$ and obtains $\alpha'$ from $\alpha$ by reducing its quantity by $\q(\beta)$. Then, it recursively calls $\matchask(B\setminus\{\beta\},A,\alpha')$ to obtain $(B',A',M')$ and finally outputs $(\hat B, \hat A, M)=(B',A',M'\cup\{m\})$. We can apply the induction hypothesis for the recursive call. Thus,
\begin{align*}
\q(A',\id(\alpha'))&=\q(\alpha')-\Q(M',\id(\alpha'))\quad\quad\quad\text{(from I.H.)}\\
\implies \q(\hat A,\id(\alpha))&=\q(\alpha')-\Q(M',\id(\alpha))\\
&\quad\quad\quad\quad\quad\text{(since $\hat A= A'$ and $\id(\alpha')=\id(\alpha)$)}\\
&= \q(\alpha) - [\q(\beta) + \Q(M',\id(\alpha))] \\ &\quad\quad\quad\quad\quad\text{(as $\q(\alpha')=\q(\alpha)-\q(\beta)$)}\\
&= \q(\alpha) - \Q(M,\id(\alpha)),
\end{align*}
where the last equality follows from observing that $M=M'\cup\{m\}$ and $m$ is a transaction involving $\alpha$ with transaction quantity $\q(\beta)$.
\end{proof}


\end{document}